\documentclass[11pt,letterpaper]{article}

\usepackage[margin=1in]{geometry}
\usepackage[T1]{fontenc}
\usepackage{lmodern}
\usepackage{microtype}

\usepackage{amsmath, amssymb, amsthm, mathtools}
\usepackage{graphicx,booktabs, enumitem, xcolor}
\usepackage[ruled]{algorithm2e}
\usepackage{hyperref}
\usepackage[nameinlink,capitalise,noabbrev]{cleveref}
\usepackage{newtxtext,newtxmath}
\usepackage{bm}
\newtheorem{theorem}{Theorem}[section]
\newtheorem{lemma}[theorem]{Lemma}

\newtheorem{corollary}[theorem]{Corollary}

\theoremstyle{definition}

\theoremstyle{remark}
\newtheorem{remark}[theorem]{Remark}

\newcommand{\R}{\mathbb{R}}

\newcommand{\Q}{\mathbb{Q}}
\newcommand{\Z}{\mathbb{Z}}

\newcommand{\Ot}{\widetilde{O}}

\newcommand{\xsol}{x^*}

\makeatletter
\renewcommand{\maketitle}{%
  \begin{center}
  {\LARGE\bfseries\@title\par}
  \ifx\@author\@empty\else\vspace{0.7em}{\large\@author\par}\fi
  \vspace{0.8em}{\normalsize\@date\par}
  \end{center}
  \vspace{0.4em}}
\makeatother

\title{\bfseries Solving Linear Systems in $\bm{{\Ot(mn \log \frac{\kappa}{\epsilon})}}$ Bit Operations}
\author{
    Jonathan A.\ Kelner\\
    MIT\\
    \texttt{kelner@mit.edu}
}

\date{}

\begin{document}

\maketitle
\begin{abstract}
We give a deterministic algorithm that solves a nonsingular linear system $Ax=b$, where $A\in\R^{n\times n}$ has $m$ nonzero entries and condition number $\kappa$, to any relative residual tolerance $0<\epsilon\le1/2$ using $\Ot(mn\log(\kappa/\epsilon))$ bit operations for inputs with logarithmically many bits per entry.
For sparse, polynomially conditioned systems with $m=\Ot(n)$, this gives an $\Ot(n^2)$ algorithm for any inverse-polynomial accuracy, improving on the algorithm of Peng and Vempala, as sharpened by Nie, whose running time in this regime is approximately $O(n^{2.2707})$ with the best current matrix multiplication exponent,
	and largely closing a gap between the idealized performance of the conjugate gradient method in exact arithmetic and the running time achievable with finite-precision computation that has persisted for over 70 years.

The algorithm is surprisingly simple. For integer inputs, we apply Dixon's lifting algorithm to the perturbed system $(A+I/R)x=b$ for a suitable integer $R$.
After scaling, the matrix of this system is $RA+I\equiv I\pmod R$, so its modular inverse is trivial, and each lifting step needs only a sparse matrix-vector product with $A$ on $O(\log R)$-bit numbers.
Fast rational reconstruction then recovers the exact solution of the perturbed system, which is an $\epsilon$-accurate solution of the original one.
Normalization and rounding extend the result to fixed-point and floating-point inputs, with floating-point outputs represented using short integer significands and a common encoded exponent.
A computable certificate removes the need for prior knowledge of $\kappa$.
\end{abstract}

\section{Introduction}

Solving systems of linear equations is a foundational problem in mathematics and computation,
with pervasive applications in both theory and practice.
There is an extensive literature aimed at developing stable and efficient algorithms for doing so for various classes of linear systems, and the study of such algorithms has led to fundamental ideas and techniques in algorithms, optimization, and numerical analysis.
Nevertheless, there remain substantial gaps in our understanding of the complexity of the problem, particularly in the case when the matrices involved are sparse.

Let $A\in\R^{n\times n}$ be an invertible matrix with $m$ nonzero entries and condition number $\kappa=\|A\|_2\|A^{-1}\|_2$, and let $b\in \R^n$.
 The choice of algorithm and best known complexity bounds for solving $Ax=b$ depend on the relationship between $m$ and $n$, as well as on whether we assume exact arithmetic or account for the bit complexity of finite-precision operations.

In exact arithmetic, classical algorithms based on Gaussian elimination and triangular factorization use $O(n^3)$ arithmetic operations.
Bunch and Hopcroft~\cite{BunchHopcroft1974} showed how to accelerate matrix factorization and inversion using fast matrix multiplication, which allows one to solve linear systems using $O(n^{\omega})$ arithmetic operations, where $\omega< 2.371177$~\cite{dupont2026improving} is an admissible matrix-multiplication exponent.
For sufficiently dense general $A$, this remains the fastest algorithm known.

Demmel, Dumitriu, and Holtz~\cite{DemmelDumitriuHoltz2007} showed that a wide range of algorithms based on fast matrix multiplication can be implemented stably in finite-precision arithmetic.
If the entries of $A$ and $b$ are represented as $B$-bit fixed-point values, their work allows one to find a vector $\widehat{x}$ such that 
$\|A\widehat x-b\|_2\le\epsilon\|b\|_2$ 
in time $\Ot\left(n^{\omega}(B+\log \epsilon^{-1}+\log \kappa)\right)$.%
\footnote{The running time is stated in~\cite{DemmelDumitriuHoltz2007} as $O(n^{\omega+\eta})$ for any $\eta>0$, but Peng and Vempala observe that the analysis in~\cite{DemmelDumitriuHoltz2007} gives $\Ot(n^\omega)$~\cite[Section~6.3, footnote~3,
p.~59]{PengVempalaArxivV2}.} 
When entries of $A$ and $b$ have logarithmically many bits and $1/\epsilon$ and $\kappa$ are polynomial in $n$, this matches the result in exact arithmetic up to polylogarithmic factors.

When $A$ is sparse, however, there has been a longstanding significant gap between the exact and finite-precision cases, dating back to 1952,
when Hestenes and Stiefel~\cite{HestenesStiefel1952} introduced the conjugate gradient method, and Lanczos published a closely related algorithm~\cite{lanczos1952solution}.
For symmetric positive definite $A$, these papers gave iterative algorithms and showed that, if all arithmetic is exact, they solve $Ax=b$ in at most $n$ iterations.	
Each iteration consists of a constant number of matrix-vector multiplications and vector operations, giving a cost of $O(m)$ arithmetic operations per iteration, and thus $O(mn)$ overall.
For a general nonsingular $A$, one can instead solve $A^T A x =A^T b$.
While $A^T A$ may no longer be sparse, one can still multiply a vector by it using $O(m)$ operations, resulting in an $O(mn)$ algorithm to solve linear systems in exact arithmetic.

However, the exact-arithmetic termination bound does not by itself yield a comparable bit-complexity bound.
Exact rational quantities can grow in size, while rounding can destroy the $n$-step termination guarantee~\cite{Greenbaum1989,MuscoMuscoSidford2018}.
Storing intermediate values exactly requires a number of bits that can grow linearly with $n$, even when the entries of $A$ and $b$ are bounded integers.
This introduces a factor of $n$ in the running time when accounting for bit complexity, making it asymptotically uncompetitive with the dense algorithms above.

A blowup in the required bit complexity arises from genuine numerical instability, not just a weakness in the analysis, and round-off errors can substantially degrade convergence in practice. 
Conjugate gradient remains efficient for well-conditioned matrices, and it has been shown to stably solve symmetric positive definite linear systems in $\Ot(\sqrt{\kappa} \log{\epsilon^{-1}})$ iterations using $O(\log(n\kappa/\epsilon))$ bits of precision for all operations when implemented appropriately.
However, there are simple instances with polynomially bounded entries and polynomial condition number that, for any $c>0$, appear to require $\Omega(n/\log n)$ bits of precision to converge in $O(n^c)$ steps~\cite{MuscoMuscoSidford2018}, assuming current upper bound techniques qualitatively capture the algorithm's actual behavior.\footnote{Current analyses of finite-precision conjugate gradient relate convergence guarantees to the existence of polynomials that are small near the eigenvalues of $A$~\cite{Greenbaum1989}.  
These bounds appear predictive of practical performance, but they are upper bounds, not lower bounds. In~\cite{MuscoMuscoSidford2018}, the authors show lower bounds on the degrees of such polynomials, not formal convergence lower bounds.}

As such, except in the presence of strong condition number bounds, there were no known algorithms that solved general sparse linear systems asymptotically faster than dense ones.
This remained the case until a 2021 breakthrough result by Peng and Vempala~\cite{PengVempala2021}.
They combined a carefully constructed randomized block Krylov method with fast matrix multiplication and a recursive matrix factorization to obtain a running time of 
\[
\Ot\!\left(
\max\left\{
n^2 m^{\frac{\omega-2}{\omega-1}},
\;
n^{\frac{5\omega-4}{\omega+1}}
\right\}
\log^2\!\left(\frac{\kappa}{\epsilon}\right)
\right)
\] 
for inputs whose entries fit in logarithmically many bits, with a supplied condition-number bound and a randomized high-probability guarantee.
In the polynomial-conditioning, inverse-polynomial-accuracy regime, a later matrix anti-concentration bound by Nie~\cite{Nie2022}
eliminated the second dimension term, giving the running-time bound
\[
\Ot\!\left(
n^2 m^{\frac{\omega-2}{\omega-1}}
\right).
\]
For sparse matrices with $m=\Ot(n)$, polynomial condition number, logarithmically many bits per entry, and inverse-polynomial $\epsilon$, this gives a running time of approximately $\Ot(n^{2.2707})$ using the current bound $\omega<2.371177$.

In this paper, we give an algorithm that solves linear systems in 
\[	
	\Ot\left(m n \log\left(\frac{\kappa}{\epsilon}\right)\right)
\] 
bit operations for inputs with logarithmically many bits per entry, 
closing much of the gap between exact and finite-precision arithmetic for sparse linear systems.
For sparse matrices with $m=\Ot(n)$ nonzero entries and polynomial condition number, this solves $Ax=b$ to any inverse-polynomial error in time $\Ot(n^2)$.

Our algorithm is surprisingly simple.
Rather than building on conjugate gradient and trying to make it stable,
we start with an algorithm by Dixon~\cite{Dixon1982} for \emph{exactly} solving linear systems with integer entries and try to make it faster.
Dixon's algorithm works by solving the linear system modulo $p^k$ for increasing powers of a prime $p$.
It starts by computing an inverse of $A$ modulo $p$ and applies it to obtain a solution modulo $p$.
Given a solution modulo $p^k$, Dixon shows one can use this inverse and a small amount of additional computation to produce a ``lifted'' solution modulo $p^{k+1}$. Starting from the solution modulo $p$, performing $L-1$ successive lifts gives a solution modulo $p^L$.
Given such a solution for sufficiently large $L$, one can use a variant of the extended Euclidean algorithm to recover an exact rational solution. 

As described, this appears to be too slow---
even the first step of finding a modular inverse 
using a general dense inversion algorithm gives
the $O(n^\omega)$ bound on which we were hoping to improve, and subsequent iterations
 involve  applying a potentially dense modular inverse.
However, we do not need an exact rational solution, just an approximate one.
This allows us to add a small perturbation to our input and solve a nearby system instead.
It turns out that we can do this in a particularly nice way that makes modular inversion trivial and allows the iterations to work with matrices with $O(m)$ nonzero entries.
Combining this with an existing faster algorithm for the rational recovery step yields our desired running time.

\subsection{Main Result}
 
For generality, we state our main result for inputs expressed in floating point.
As we show in~\cref{sec:floating_point}, this implies the analogous result for fixed-point inputs as a simple corollary.

We represent floating-point numbers as $u2^e$, where $u,e\in\Z$ are stored in binary, and we write $\operatorname{bits}(w)=1+\lceil\log_2(|w|+1)\rceil$ for the signed binary length of an integer, which we use to quantify the complexity of $u$ and $e$.
The matrix is supplied as a list of its nonzero entries and their row and column indices, and the right-hand side is supplied as a vector of encoded entries.
Throughout the paper, we write $\kappa_2(A)=\|A\|_2\|A^{-1}\|_2$, or simply $\kappa$ when $A$ is clear, logarithms are base $2$,
and the notation $\Ot$ suppresses factors polylogarithmic in the dimension and bit-length, not polynomial factors in the numerical magnitudes of encoded exponents.

\begin{theorem}[Main] \label{thm:main}
Let $A\in\R^{n\times n}$ be nonsingular with $m$ nonzero entries, and let $b\in\R^n$.
Suppose every entry is supplied exactly as $u2^e$ with
$\operatorname{bits}(u)\le B_{\rm sig}$ and $\operatorname{bits}(e)\le B_{\rm exp}$.
Let $\epsilon=2^{-E}$ for an integer $E\ge1$, and let
\[
    \rho=1+E+\lceil\log_2\kappa_2(A)\rceil+\lceil\log_2(n+1)\rceil.
\]
There is a deterministic algorithm that returns integers $q_1,\ldots,q_n,\gamma$ such that
\[
    \widehat x_i=q_i2^\gamma,
    \qquad \|A\widehat x-b\|_2\le\epsilon\|b\|_2.
\]
Its running time and storage are, respectively,
\begin{align*}
    &\Ot\!\left((m+n)(B_{\rm sig}+B_{\rm exp})+mn\rho\right)
        &&\text{bit operations},\\
    &\Ot\!\left((m+n)(B_{\rm sig}+B_{\rm exp})+n^2\rho\right)
        &&\text{bits}.
\end{align*}
Each $q_i$ has $O(\rho)$ bits, and the common exponent $\gamma$ has
$O(B_{\rm exp}+\log(B_{\rm sig}+\rho+1))$ bits.
The algorithm does not require a supplied condition-number bound.
\end{theorem}

With polylogarithmic encodings, polynomial conditioning, and inverse-polynomial accuracy, the running time is $\Ot(mn)$. If $A$ is sparse with $m=\Ot(n)$ nonzero entries, this becomes $\Ot(n^2)$.

The core of our algorithm is an algorithm for the case where the entries of $A$ and $b$ are integers, which we present in~\cref{sec:integer_solver}.
In~\cref{sec:floating_point}, we prove~\cref{thm:main} by showing how to normalize and round floating-point inputs to obtain an integer system, which we solve using an algorithm from~\cref{sec:integer_solver}.
Fixed-point inputs are then a straightforward specialization, which we discuss briefly after the proof.



\section{Background on Dixon's Algorithm and Rational Reconstruction}
Our algorithm is based on Dixon's algorithm~\cite{Dixon1982} for finding the exact rational solution to a linear system with integer entries.
It works by finding a $p$-adic approximation of the solution to a sufficiently high precision to uniquely determine the rational coordinates of the solution vector, and then using an algorithm based on the extended Euclidean algorithm to recover the solution. 
For completeness, we review it here.

\subsection{Dixon's Algorithm for Exactly Solving Systems with Integer Entries }
Dixon's algorithm finds the exact rational solution to $M x =b$ for an invertible matrix $M\in \Z^{n \times n}$ and a vector $b\in \Z^n$. 	 
Dixon's original paper worked modulo a prime $p$, 
whereas we work modulo any integer $R\geq 2$ with $\gcd(R, \det(M))=1$.
This does not require any substantive changes to Dixon's original arguments. 
Otherwise, the presentation and proofs below follow those in~\cite{Dixon1982}. 
The algorithm has three main steps:

\paragraph{Step 1: Invert $\bm{M}$ mod $\bm{R}$:} Use an appropriate matrix inversion algorithm to find an integer matrix $C$ such that $M C \equiv I \pmod R$, with the entries of $C$ in $\{0,\dots,R-1\}$.

\paragraph{Step 2: Use successive lifting to solve $\bm{Mx \equiv b \pmod{R^L}}$ for sufficiently large $L$:}
We iteratively construct a sequence of vectors 
	$\{x_i\}_{i=0}^{L-1}$ with entries in $\{0,\dots,R-1\}$
	such that the sums  
	\[
		y_k:=\sum_{i=0}^{k-1} R^i x_i
	\]
 	obey
	\begin{equation}\label{eq:lifting_eq}
		M y_k \equiv b \pmod{R^k},
	\end{equation}
as well as an auxiliary sequence of vectors $\{b_i\}_{i=0}^{L}$ that keep track of the residual at each step.
To do so, we initialize $b_0:=b$ and compute $\{x_i\}_{i\geq 0}$ and $\{b_i\}_{i\geq 1}$ using the recurrence:
\begin{align}
	x_i &:= C b_i \bmod R \label{eq:x_update}\\
	b_{i+1}&:=\frac{b_i - M x_i}{R}\label{eq:b_update}.
\end{align}
Assuming $b_i\in\Z^n$, Equation~\eqref{eq:x_update} gives $M x_i \equiv M C b_i \equiv b_i \pmod{R}$, so the division in~\eqref{eq:b_update} is exact and $b_{i+1}\in\Z^n$.

To see that these vectors obey~(\ref{eq:lifting_eq}), note that rearranging the definition of $b_{i+1}$ gives $M x_i = b_i - R b_{i+1}$, so
$$M y_k=\sum_{i=0}^{k-1}R^i M x_i
=\sum_{i=0}^{k-1} R^i (b_i - R b_{i+1})=b_0-R^k b_k \equiv b \pmod{R^k}$$
as claimed.

During lifting, the algorithm stores the $x_i$ but does not repeatedly assemble the growing prefixes $y_k$.
The final vector $y_L$ is assembled once for reconstruction.  When $R=2^s$, this is a concatenation of $s$-bit blocks and costs $O(nLs)$ bit operations.

\paragraph{Step 3: Reconstruct the rational solution from $\bm{y_L}$:} Let $\xsol \in \Q^n$ be the exact rational solution to $Mx=b$, and write its $i$th coordinate as a ratio of integers $\xsol_i=f_i/g_i$ with $\gcd(f_i,g_i)=1$.
We have $\gcd(g_i, R)=1$ by Cramer's rule and the assumption that $\gcd(R,\det(M))=1$, so $g_i$ is invertible in $\Z/R^k\Z$ for any $k\geq 1$.
This allows us to map $f_i/g_i$ into $\Z/R^k\Z$ by multiplying $f_i$ by the inverse  of $g_i$.
We refer to the result as $\xsol_i \bmod R^k$, and it is straightforward to check that the resulting vector $\xsol \bmod R^k$ is the solution to  $Mx\equiv b \pmod{R^k}$. 

This implies that $\xsol \equiv y_L \pmod{R^L}$, where $y_L$ is the vector constructed in Step 2, and our goal is to reconstruct $\xsol$ from $y_L$.
In general, multiple rational numbers map to the same values modulo $R^L$.
However, if we assume bounds on the entries of $M$ and $b$, we can prove bounds on $f_i$ and $g_i$. 
Dixon showed that, for sufficiently large $L$, there is a unique rational vector equivalent to $y_L$ that obeys these bounds, and he showed how to find it efficiently using a variant of the extended 	Euclidean algorithm.
We formalize the conditions for unique reconstruction in the following theorem, which slightly modifies Dixon's original guarantees to 
put them in a form that is easier to apply in our setting.

\begin{theorem}\label{thm:reconstruction}
Let $M\in \Z^{n\times n}$ be invertible, let $b\in \Z^n$, and 
let $R\geq 2$ and $L\geq1$ be integers such that $\gcd(R,\det M)=1$. 
Suppose $M$ and $b$ satisfy entry-wise bounds
\[
  |M_{ij}| \leq T, \qquad |b_i| \leq T 
\] 
for some $T\geq 1$ and for all $i$ and $j$, and let $H$ be an integer with
\[
	H \geq n^{n/2} T^n.
\] 
Suppose that $v\in \Z^n$ satisfies
\[
	M v \equiv b \pmod{R^L},
\]
and that
\begin{align}
R^L > 2 H^2. \label{eq:L_bound}
\end{align}
Then $\xsol=M^{-1}b$ is the unique vector in $\Q^n$ congruent to $v$ modulo $R^L$ such that, when its coordinates are written as fractions $a_i/d_i$ in lowest terms, 
\[
	|a_i|\leq H, \qquad 1\leq d_i \leq H
\] 
for all $i$, where congruence of a rational number modulo $R^L$ is defined only when its reduced denominator is coprime to $R^L$.
\end{theorem}
\begin{proof}
Let $M_i$ be the matrix obtained from $M$ by replacing its $i$th column with $b$,
and write each coordinate $\xsol_i=f_i/g_i$ as a fraction in lowest terms with $g_i>0$. 
By Cramer's rule, $\xsol_i=\det M_i/{\det M}$. 
The columns of $M$ and $M_i$ all have norm at most $\sqrt{n}T$, so 
Hadamard's inequality gives
\[
    |{\det M}|,\ |{\det M_i}|
    \le n^{n/2}T^n\le H
\]
for all $i$, and thus $|f_i| \leq H$ and $1\leq g_i\leq H$, as required. 

Since $g_i$ divides $\det M$, and $\gcd(R,\det M)=1$ by assumption, we have $\gcd(R^L ,g_i)=1$,
so it makes sense to talk about $\xsol_i \bmod{R^L}$.
Since
\[
	M\xsol \equiv b \equiv Mv \pmod{R^L},
\]
and $M$ is invertible modulo $R^L$, we have $\xsol \equiv v \pmod{R^L}$.  
The vector $\xsol$ therefore satisfies the stated conditions.

For uniqueness, let $y$ be any vector meeting these conditions, and write its coordinates as fractions
$y_i=a_i/d_i$ in lowest terms.
Since $y_i \equiv \xsol_i \pmod{R^L}$, we have 
\[
R^L \mid (a_i g_i - f_i d_i). 
\] 
However, our assumption in~\eqref{eq:L_bound} implies that
\[
	\left| a_i g_i - f_i d_i\right|
	\leq  |a_i| g_i + |f_i| d_i
	\leq 2H^2
	<R^L,	
\]
so we must have $a_i g_i - f_i d_i=0$, and thus  $y_i =\xsol_i$ for all $i$.
\end{proof}

\subsection{Running Time of Dixon's Algorithm}
The following lemma allows us to bound the bit complexity of the quantities computed by Dixon's algorithm. Here $\|M\|_\infty=\max_i\sum_j|M_{ij}|$ is the maximum absolute row sum:
\begin{lemma}\label{lemma:residual_bound}
	The residual vectors computed during Step 2 of Dixon's algorithm have
	\[
	\|b_j\|_{\infty} \leq W:= \max\left(\|b\|_{\infty},\|M\|_{\infty} \right).
	\]
\end{lemma}
\begin{proof}
	We proceed by induction.
	The initial vector $b_0=b$ clearly obeys the desired bound.
	Now assume that $\|b_i\|_\infty \leq W$. The coordinates of $x_i$ are in $\{0,\dots,R-1\}$, so $\|x_i\|_{\infty}\leq R-1$, and therefore
	\[
	 	\|b_{i+1}\|_{\infty}=\left\|\frac{b_i - M x_i}{R}\right\|_{\infty} 
	 	\leq \frac{1}{R} \left(\|b_i\|_{\infty}+\|M\|_{\infty}\cdot\|x_i\|_{\infty} \right)
	 	\leq \frac{1}{R}\Big(W +(R-1)W \Big)=W,
	\]	
	so the desired result follows by induction.
\end{proof}

Suppose for simplicity that the entries of $M$ and $b$ are polynomially bounded $O(\log n)$-bit integers and that $R$ is a polynomially bounded prime.
Straightforward implementations of all three steps of Dixon's algorithm do not, in general, achieve the desired $\Ot(mn)$ bound:
\begin{itemize}
	\item \textbf{Step 1:} Inverting $M$ mod $R$ in Step 1 takes $\Ot(n^3)$ time using classical Gaussian elimination, or $\Ot(n^{\omega})$ using inversion algorithms based on fast matrix multiplication, where $\omega$ is an admissible matrix-multiplication exponent as above.
	\item \textbf{Step 2:} \cref{thm:reconstruction} gives a bound of $\Ot(n)$ for the number of iterations $L$.
	The entries of each $x_i$ and $b_i$ are polynomially bounded, and thus all of the arithmetic operations in Step 2 require only logarithmic precision.  
	
	However, even if $M$ is sparse, $C$ can be dense, so each iteration could require a quadratic number of arithmetic operations.  This gives a total running time of $\Ot(n^2 L)=\Ot(n^3)$ for Step~2.  
	\item \textbf{Step 3:} One can show that recovering each coordinate with the Euclidean algorithm modulo $R^L$ takes time $\Ot(L^2)=\Ot(n^2),$ so the total running time for Step 3 is also $\Ot(n^3)$. 	
\end{itemize}

Our algorithm will apply Dixon's algorithm in a setting where the first two steps can be performed more efficiently.
For Step 3, we will replace the classical Euclidean reconstruction algorithm with a more recent algorithm that only requires $\Ot(L)$ time  per coordinate, allowing us to implement Step 3 in time $\Ot(n^2)$.

\subsection{Faster Reconstruction}
In this section, we state the faster rational reconstruction results we'll use in our algorithm. 
\begin{theorem}[Fast rational reconstruction]\label{thm:scalar_reconstruction}
Let $Q,v,N,D$ be integers with \(Q\ge 2\), \(0\le v<Q\), and \(N,D\ge1\) such that
$$
    Q>2ND.
$$
There is at most one reduced fraction \(a/c\) satisfying
$$
    |a|\le N,\qquad 1\le c\le D,\qquad
    \gcd(c,Q)=1,\qquad
    a\equiv vc\pmod Q.
$$
If such a fraction exists, it can be recovered deterministically in
$$
    O(\mathsf M(s)\log s)
$$
bit operations, where
$$
    s=\left\lceil\log_2(Q+1)\right\rceil
$$
and \(\mathsf M(s)\) is the bit complexity of multiplying two \(s\)-bit integers.
In particular, $\mathsf M(s)=\Ot(s)$, and the assumption $Q>2ND$ implies that $N,D<Q$, so all reconstruction parameters have $O(s)$ bits;
hence the running time is
$$  
    \widetilde O(\log Q).
$$
\end{theorem}
Uniqueness is straightforward: if $a/c$ and $a'/c'$ both satisfy these conditions, then $ac'\equiv vcc'\equiv a'c\pmod Q$ and $|ac'-a'c|\le 2ND<Q$, so $ac'=a'c$. 
Since both fractions are reduced with positive denominators, they are equal.
We use the classical bounded rational reconstruction theorem of Wang, Guy, and Davenport~\cite{WangGuyDavenport1982}; Collins and Encarnaci\'on~\cite{CollinsEncarnacion1995} later corrected an error in Wang's treatment and gave an efficient implementation. 
Both they and Dixon gave reconstruction algorithms based on the extended Euclidean 
algorithm. 
Concretely, one runs the extended Euclidean algorithm on $Q$ and $v$, stops at the first remainder $r_i\le N$, whose cofactor $t_i$ satisfies $r_i\equiv t_iv\pmod Q$, and returns $r_i/t_i$ (with the sign normalized so that the denominator is positive) provided $|t_i|\le D$ and $\gcd(t_i,Q)=1$; when $Q>2ND$, this returns the fraction whenever it exists.
See Monagan~\cite[Theorem~1 and Algorithm~RR, pp.~244--245]{Monagan2004} for the precise bounded-reconstruction criterion and extended-Euclidean stopping rule used above.
The stated bit bound follows by implementing this partial extended-Euclidean computation with fast arithmetic; see Wang and Pan~\cite{WangPan2003}.
Pan and Wang~\cite[p.~502]{PanWang2004} state the $O(\mathsf M(s)\log s)$ bound explicitly and explain an alternative derivation using fast Euclidean computation and a product tree.
The modulus $Q$ need not be prime; the required coprimality condition is $\gcd(c,Q)=1$.

\begin{corollary}\label{cor:vector_reconstruction}
Under the hypotheses of Theorem~\ref{thm:reconstruction}, represent each $v_i$ in $\{0,\ldots,R^L-1\}$. There is a function
\textnormal{\textsc{Recover}}$(v,R,L,H)$ that recovers the exact solution
$$
    \xsol=M^{-1}b
$$
by applying scalar rational reconstruction to each coordinate, in
$$
    O\!\left(
        n\,\mathsf M(L\log R)\log(L\log R)
      \right)
    =
    \widetilde O(nL\log R)
$$
bit operations.  In particular, if \(L=O(n)\), the reconstruction cost is
$$
    \widetilde O(n^2\log R).
$$
\end{corollary}

\begin{proof}
By Theorem~\ref{thm:reconstruction}, each coordinate
\(\xsol_i=a_i/d_i\) satisfies
$$
    |a_i|,d_i\le H,\qquad
    a_i\equiv v_i d_i\pmod{R^L},
$$
and \(\gcd(d_i,R^L)=1\).  Since \(R^L>2H^2\),
Theorem~\ref{thm:scalar_reconstruction} applies to each coordinate with
\(N=D=H\).  As \(\log(R^L)=L\log R\), reconstructing all \(n\) coordinates
gives the stated bound.
\end{proof}

\section{Approximately Solving Systems with Integer Entries}\label{sec:integer_solver}
In this section, we give a solver for inputs whose entries are integers. 
We show how to use this to solve floating-point inputs in~\cref{sec:floating_point}.
   
\begin{theorem}[Solver for integer inputs] \label{thm:integer_solver}
Let $A\in \Z^{n \times n}$ be an invertible matrix with $m$ nonzero entries and condition number $\kappa$, let $b\in \Z^n$ be a vector, and let 
$T_A  = \max\left(1,\max_{i,j} |A_{ij}|, \max_{i}|b_i|\right)$.
Assume that $T_A < 2^{B-1}$ and that the entries are given in a signed binary representation with at most $B$ bits per entry.
Given an integer $E\ge0$, put $\epsilon=2^{-E}$. There is a deterministic algorithm that finds a vector $x\in \Q^n$ such that 
\[
\|Ax-b\|_2 \leq \epsilon \|b\|_2
\]
using at most $\Ot(m n \tau)$ bit operations and $\Ot(n^2\tau)$ bits of space, where
$\tau = \max(B, \log \kappa, \log 1/\epsilon)$,
and the output $x$ is a rational vector whose coordinates have numerators and denominators of $\Ot(n\tau)$ bits.
\end{theorem}

To get the claimed $\Ot(m n \tau)$ algorithm, we solve a related linear system for which we can give a very fast implementation of Dixon's algorithm.
Let $R=2^s$ for a sufficiently large $s$ to be chosen later, and let
\[M=R A + I.\]
  We'll solve the system $M z =b$ and then return $\widehat x = R z$.  Note that
$$M z = (R A + I) z   =(A + I/R ) \widehat x =b,$$
so $$A \widehat x - b = -\widehat x /R.$$
If we choose a sufficiently large $R$, the residual will be small, and thus $\widehat x$ will give the required approximate solution.

The key point is that $M \equiv I \pmod{R}$. 
Since $M/R=A+I/R$, the scaled working matrix differs from $A$ by only $I/R$, while its modular structure lets us implement Dixon's algorithm much more efficiently.
In particular, the inverse of $M$ modulo $R$ is $I$, so no modular-inverse computation or dense modular-inverse application is needed.
This makes the first step of Dixon's algorithm trivial, and it removes the dense matrices from Step 2 so that all matrix-vector multiplications involve matrices with $O(m)$ nonzero entries. 
Viewed differently, since $RA\equiv 0\pmod R$, the series $z=(I+RA)^{-1}b=\sum_{k\ge0}(-R)^kA^kb$ converges $R$-adically, and the lifting below computes the base-$R$ digits of its $R$-adic expansion using only products with $A$, while, for the choice of $R$ below, Lemma~\ref{lemma:residual_bound} keeps intermediate quantities in the iterations to $O(\log R)$ bits.

In Section~\ref{sec:solve_perturbed}, we'll specialize Dixon's algorithm to systems of the form above and analyze its running time.
In Section~\ref{sec:perturbation_error}, we'll prove~\cref{thm:integer_solver} by bounding the error to show that the solution to the modified system yields our desired approximate solution to the original linear system.

\begin{remark}
The choice of $R$ to be a power of two, rather than a prime as in Dixon's original paper, is not essential for the recurrence or stated overall running time, but it helps in two ways:
\begin{enumerate}
\item It simplifies the implementation, as assembling $y_L=\sum_{i=0}^{L-1} R^i x_i$ in binary just involves concatenating the stored digits coordinatewise.
For a general integer $R\ge2$, the same assembly can instead
be performed by divide-and-conquer base conversion in
$\Ot(nL\log R)$ bit operations; see
Brent and Zimmermann~\cite[Section~1.7.2, Algorithm~1.25]{BrentZimmermann2010}.
In addition, the modular reduction and division steps become bitwise operations that are likely to be faster in practice. 
\item It makes modulus selection immediate: for an integer threshold $X\ge2$, take $R=2^{\lceil\log_2 X\rceil}$, so $X\le R<2X$. Requiring a prime would introduce a separate prime-selection step, and finding a prime of appropriate size deterministically is nontrivial.
Our construction avoids that step. 
\end{enumerate}

\end{remark}

\subsection{Dixon's Algorithm for \texorpdfstring{$\bm{M\equiv I \pmod{R}}$}{M congruent to I modulo R}}\label{sec:solve_perturbed}

For this subsection, let $A\in\Z^{n\times n}$ have $m$ nonzero entries, let $b\in\Z^n$, and let
\[
    T_A=\max\{1,\max_{i,j}|A_{ij}|,\max_i|b_i|\}.
\]
Let $M=RA+I$ with $R=2^s\geq2nT_A$ and $s\ge1$. Since $\det(M)\equiv1\pmod R$, $M$ is nonsingular and $\gcd(R,\det(M))=1$, so $M$ and $R$ meet the requirements of~\cref{thm:reconstruction}.
Let 
\[
	T_M  = \max\left(\max_{i,j} |M_{ij}|, \max_{i}|b_i|\right)\leq RT_A+1.
\]
The reconstruction guarantees in the theorem require $R^L > 2 H^2$ for an integer height bound
$H \ge n^{n/2} T_M^n$.

The assumption $R\geq 2n T_A$ implies that $n \leq R$ and 
\[
	T_M \leq R T_A+1 \leq \frac{R^2}{2n}+1 \leq R^2,
\]
so
\[
	n^{n/2} T_M^n
	\leq 
	R^{n/2} R^{2n}
	=
	R^{5n/2}
	\leq 
	R^{3n},
\]
and thus $H=R^{3n}$ is an admissible height bound, with $2H^2= 2R^{6n}$.
If we choose $L=6n+2$, we'll then have
\[
R^L = R^{6n+2}\geq 4R^{6n} > 2H^2,	
\]
as required. We use this height bound $H=R^{3n}$ in the reconstruction routine.

For $M=RA+I$, the recurrence in Dixon's algorithm simplifies significantly.  
Since $M^{-1}\equiv I \pmod{R}$, Equation~\eqref{eq:x_update} becomes 
\[
x_i := b_i \bmod R,
\]
and plugging $M=RA+I$ into~\eqref{eq:b_update} gives
\[
	b_{i+1}:=\frac{b_i -x_i}{R}-Ax_i.
\]
This gives us the following algorithm for exactly solving $Mz=b$: 
\begin{algorithm}[ht]
\caption{\textsc{SolvePerturbed}$(A,b,R)$}
\label{alg:identity_lift}
\LinesNumbered
\KwIn{$A\in\Z^{n\times n}$, $b \in\Z^n$, $T_A=\max\{1,\max_{i,j}|A_{ij}|,\max_i|b_i|\}$, and $R=2^s\ge2nT_A$ for an integer $s\ge1$}
\KwOut{The exact rational vector $z=(I+RA)^{-1}b$}

$L\gets6n+2$; $H\gets R^{3n}$; 
$b_0\gets b$\;
\For{$i=0,\ldots,L-1$}{
  $x_i\gets b_i\bmod R$ in $\{0,\ldots,R-1\}^n$\;\label{line:digit}
  $b_{i+1}\gets(b_i-x_i)/R-Ax_i$\;\label{line:residual}
}

Assemble $y_L=\sum_{i=0}^{L-1}R^i x_i$ by concatenating the stored digits\;\label{line:assemble}
$z \gets$ \textsc{Recover}$(y_L,R,L,H)$\tcp*{See Corollary~\ref{cor:vector_reconstruction}}\label{line:recover} 
\Return $z$\;
\end{algorithm}
\begin{theorem}\label{thm:solve_perturbed}
Let $A\in\Z^{n\times n}$ have $m$ nonzero entries, let $b\in\Z^n$, and let
\[
    T_A=\max\{1,\max_{i,j}|A_{ij}|,\max_i|b_i|\}.
\]
For $R=2^s\ge2nT_A$, with integer $s\ge1$,
\textnormal{\textsc{SolvePerturbed}}$(A,b,R)$ returns $(I+RA)^{-1}b$ and can be implemented to use $\Ot((mn+n^2)s)$ bit operations and  $\Ot((m+n^2)s)$ bits of storage.
Each reduced output numerator and denominator has $O(ns)$ bits. No nonsingularity assumption on $A$ is required.
\end{theorem}
\begin{proof}
	\textsc{SolvePerturbed} is Dixon's algorithm for $M=I+RA$ with $C=I$. Since $MI=M\equiv I\pmod R$ and $R\ge2$, the matrix $C=I$ meets the requirement of Step~1, and with this choice~\eqref{eq:x_update} and~\eqref{eq:b_update} become Lines~\ref{line:digit} and~\ref{line:residual}. The argument of Step~2 therefore gives $My_L\equiv b\pmod{R^L}$, and the entries of $y_L$ lie in $\{0,\ldots,R^L-1\}$. The discussion above shows that $\gcd(R,\det M)=1$, that the entries of $M$ and $b$ are bounded in absolute value by $T_M$, that $H=R^{3n}\ge n^{n/2}T_M^n$, and that $R^L>2H^2$. Hence~\cref{thm:reconstruction} applies with $T=T_M$, $v=y_L$, and this $H$, and by Corollary~\ref{cor:vector_reconstruction}, Line~\ref{line:recover} returns $M^{-1}b=(I+RA)^{-1}b$.
	
	Each iteration of the recurrence involves $O(n)$ scalar operations and matrix-vector multiplication with the matrix $A$, which has $m$ nonzero entries, so there are $O(m+n)$ arithmetic operations per iteration.
	The absolute values of the entries of $A$, $b$, and $x_i$ are at most $R$ by assumption, and the entries of $b_i$ have absolute value at most $\max\left(\|b\|_{\infty},\|I+RA\|_{\infty} \right)=O(R^2)$ by Lemma~\ref{lemma:residual_bound}, so they all fit in $O(s)$ bits.
	Every partial row sum in $Ax_i$ has absolute value at most $(R-1)nT_A\le R^2/2$, so these intermediate values also fit in $O(s)$ bits.
		Each arithmetic operation can thus be implemented using $\Ot(s)$ bit operations, 
		so we need $\Ot((m+n)s)$ bit operations per iteration.
		There are $L=O(n)$ iterations, so the total number of bit operations for the iterative
		portion of the algorithm is $\Ot((mn+n^2)s)$.

	Assembling $y_L$ once costs $O(nLs)=O(n^2s)$ bit operations.
	\textsc{Recover} requires $\Ot(n^2 \log R)=\Ot(n^2 s)$ bit operations by Corollary~\ref{cor:vector_reconstruction}, so the total number of bit operations is $\Ot(mns+n^2 s)=\Ot(( mn+n^2) s)$, as claimed.
	
	Each vector in the recurrence requires $O(ns)$ bits, so storing all $O(n)$ of them uses $O(n^2 s)$ bits. Each coordinate output by the reconstruction algorithm is a ratio of integers with $O(ns)$ bits, so the full output is $O(n^2s)$ bits.  
	The reconstruction algorithm computes each coordinate separately and uses at most $\Ot(ns)$ bits of space, bounded by its running time.
	Adding these to the $O((m+n)s)$ bits required to store $A$, $b$, and their indices gives a total space bound of $\Ot((m+n^2)s)$, as claimed.
\end{proof}

\subsection{Bounding the Approximation Error}\label{sec:perturbation_error}
To prove~\cref{thm:integer_solver}, we need to bound the error introduced by approximating the solution to $Ax=b$ by $Rz$ for $z=(I+RA)^{-1}b$. 
The integer solver returns rational coordinates. Lemma~\ref{lem:rounding_certificate} controls input and output rounding, and Section~\ref{sec:floating_point} obtains the short significands and common encoded output scale promised by Theorem~\ref{thm:main}.

\begin{proof}[Proof of~\cref{thm:integer_solver}]
If $b=0$ or $E=0$, return $0$. Otherwise $E\ge1$ and $0<\epsilon=2^{-E}\le1/2$.
First suppose that $\kappa$ is known, let $R=2^s$ be the smallest power of 2 greater than $\max\left(\,2nT_A,\frac{2\kappa}{\epsilon}\right)$,
and return
\[
\widehat x = R \cdot \textsc{SolvePerturbed}(A,b,R) = R(I+RA)^{-1}b
=\left(A+\frac{I}{R}\right)^{-1}b
\]
as the approximate solution.

We have $s=O(B+\log(n+1)+\log\kappa+\log(1/\epsilon))=\Ot(\tau)$, where the $\log(n+1)$ term is absorbed by $\Ot$ because $\tau\ge B\ge2$. Theorem~\ref{thm:solve_perturbed} therefore gives the claimed work and storage bounds, and scaling its output by $R$ takes $O(n^2s)$ bit operations. It remains to show that
	$\|A\widehat x-b\|_2 \leq \epsilon \|b\|_2$.
We have
\[
	\left(A+\frac{I}{R}\right) \widehat x= A\widehat x + \frac{\widehat x}{R}= b,
\]
so
\begin{equation}\label{eq:residual}
	A \widehat x -b=-\frac{\widehat x}{R}=-(I+RA)^{-1}b. 
\end{equation}
	
We can write the condition number of $A$ as the ratio of its largest and smallest singular values, $\kappa=\sigma_{\max}(A)/\sigma_{\min}(A)$.
$A$ is a nonzero matrix with integer entries, so  $\sigma_{\max}(A)\geq 1$, and therefore $\sigma_{\min}(A)	 \geq 1/\kappa$.
It follows that, for any vector $w$, we have 
\[
	\|(I+RA)w\|_2 \geq R\|Aw\|_2-\|w\|_2\geq (R\,\sigma_{\min}(A)-1)\|w\|_2
	\geq (R/\kappa-1)\|w\|_2.
\]
Applying $(I+RA)$ to both sides of \cref{eq:residual} gives $-b=(I+RA)(A\widehat x-b)$, and hence
\begin{equation}\label{eq:mult_error}
	\|b\|_2\geq (R/\kappa-1)\left\|A \widehat x -b\right\|_2.  
\end{equation}
We chose $R\geq 2\kappa/\epsilon$, so 
\[
	\frac{R}{\kappa}-1 \geq \frac{2}{\epsilon}-1=\frac{2-\epsilon}{\epsilon}\geq \frac{1}{\epsilon}.
\]
Plugging this into~\cref{eq:mult_error} and rearranging gives
\[
	\left\|A \widehat x -b\right\|_2\leq \epsilon\|b\|_2,
\]
as desired.

To remove the assumption that $\kappa$ is known, start with $s_0=\max\{\lceil\log_2(2nT_A)\rceil,E\}$ and try $s=2^j s_0$ for $j=0,1,\ldots$.
At each stage, compute $z=\textsc{SolvePerturbed}(A,b,2^s)$ and return $\widehat x=2^s z$ only if
\[
    \sum_i z_i^2\le\epsilon^2\sum_i b_i^2.
\]
This test certifies the original residual exactly because $A\widehat x-b=-z$.
The test fits in $\Ot(n^2s)$ bit operations: square the $n$ fractions of $O(ns)$ bits and sum them using a balanced binary tree. At a level combining $2^t$ terms per node, each numerator and denominator has $O(2^t ns)$ bits, so the total work per level is $\Ot(n^2s)$. Exact comparison with $2^{-2E}\sum_i b_i^2$ has the same bound since $s\ge E$.
Every stage is defined since $I+2^s A$ is invertible, and the first stage with $2^s\ge2\kappa/\epsilon$ must pass. Its $s$ is $O(B+\log(n+1)+\log\kappa+\log(1/\epsilon))$.
Since the tested bit lengths double in each stage, their sum is less than twice the last, and the number of stages is logarithmic in this last bit length. Thus the total work and storage remain $\Ot(mn\tau)$ and $\Ot(n^2\tau)$, respectively, and each returned rational coordinate has $\Ot(n\tau)$ bits.
\end{proof}

\subsection{Input and Output Rounding}\label{sec:rounding_certificate}
The exact integer solver also gives a useful certificate for a system whose entries have been rounded that we will use in~\cref{sec:floating_point}.
The following observation separates this error estimate from the input representation.
\begin{lemma}[Rounding certificate]\label{lem:rounding_certificate}
Let $A_0\in\R^{n\times n}$ and $b_0\in\R^n$ satisfy
\[
    \tfrac12\le\max_{i,j}|(A_0)_{ij}|\le1,
    \qquad \tfrac12\le\|b_0\|_\infty\le1.
\]
For an integer $h\ge1$, put $\ell=\lceil\log_2n\rceil$, $S=2^h$, and
\[
    C=\operatorname{round}(SA_0),\qquad
    d=\operatorname{round}(Sb_0),\qquad
    R=2^{h+\ell+1},\qquad y=R(I+RC)^{-1}d.
\]
Round entries to nearest with a fixed tie rule, preserving zero entries.
For a power of two $P\ge1$, let $\widehat y=P^{-1}\operatorname{round}(Py)$. Then
\begin{equation}\label{eq:rounding_certificate_bound}
    \frac{\|A_0\widehat y-b_0\|_2}{\|b_0\|_2}
    \le \frac{2n^2(\|y\|_\infty+1)}{S}+\frac{n^2}{P}.
\end{equation}
In particular, for $0<\epsilon\le1$ and $P\ge4n^2/\epsilon$, the test
\begin{equation}\label{eq:rounding_certificate_test}
    8n^2(\|y\|_\infty+1)\le\epsilon S
\end{equation}
certifies relative residual at most $\epsilon/2$.
If $A_0$ is nonsingular, the test holds for $S\ge64\kappa_2(A_0)n^3/\epsilon$.
\end{lemma}
\begin{proof}
Let $\widetilde A=C/S+I/(RS)$ and $\widetilde b=d/S$, so $\widetilde A y=\widetilde b$.
The rounded entries obey $|C_{ij}|,|d_i|\le S$, and $R\ge2nS$; thus Theorem~\ref{thm:solve_perturbed} applies even if $C$ is singular.
Entrywise rounding and the Frobenius norm give
\[
    \|\widetilde A-A_0\|_2\le\frac{n}{2S}+\frac1{RS}\le\frac nS,
    \qquad \|\widetilde b-b_0\|_2\le\frac{\sqrt n}{2S}.
\]
Using $A_0y-b_0=(A_0-\widetilde A)y+(\widetilde b-b_0)$,
$\|y\|_2\le\sqrt n\|y\|_\infty$, and $\|b_0\|_2\ge1/2$, the unrounded relative residual is at most
$(2n\sqrt n\|y\|_\infty+\sqrt n)/S\le2n^2(\|y\|_\infty+1)/S$.
Since $\|A_0\|_2\le n$, output rounding contributes at most $n\sqrt n/P\le n^2/P$.
This proves~\eqref{eq:rounding_certificate_bound} and the certificate.

For the final assertion, write $\kappa_0=\kappa_2(A_0)$.
An entry of magnitude at least $1/2$ implies $\sigma_{\min}(A_0)\ge1/(2\kappa_0)$.
Once $S\ge4\kappa_0n$, the preceding perturbation bound gives
$\sigma_{\min}(\widetilde A)\ge1/(4\kappa_0)$, and hence
$\|y\|_\infty\le\|y\|_2\le4\kappa_0\|d/S\|_2\le4\kappa_0\sqrt n$.
Thus $\|y\|_\infty+1\le5\kappa_0n$, and
$S\ge64\kappa_0n^3/\epsilon$ implies~\eqref{eq:rounding_certificate_test}.
\end{proof}

\section{Floating-Point Inputs and Proof of the Main Theorem}\label{sec:floating_point}
We now describe the conversion from floating-point inputs to an integer system we can pass to \textsc{SolvePerturbed}.
There are three separate operations: normalize the input scales, round the normalized entries to short integers, and convert the exact working solution back to the original scale.
We first explain how to choose the rounding precision from a supplied condition-number bound, and then replace that choice by a certified precision search.

\subsection{Normalization and conversion to integers}\label{sec:fp_conversion}
Return $(q,\gamma)=(0,0)$ if $b=0$; henceforth assume $b\ne0$.
For a nonzero integer $u$, let $\operatorname{len}(u)=\lfloor\log_2|u|\rfloor+1$.
An input entry $u2^e$ has magnitude in $[2^{e+\operatorname{len}(u)-1},2^{e+\operatorname{len}(u)})$.
Let $\alpha$ and $\beta$ be the maxima of $e+\operatorname{len}(u)$ over the nonzero entries of $A$ and $b$, respectively, and define
\[
    A_0=2^{-\alpha}A,\qquad b_0=2^{-\beta}b.
\]
This normalization is performed by subtracting exponents, not by expanding $2^\alpha$ or $2^\beta$.
It gives
\[
    \tfrac12\le\max_{i,j}|(A_0)_{ij}|<1,\qquad
    \tfrac12\le\|b_0\|_\infty<1,
    \qquad \kappa_2(A_0)=\kappa_2(A).
\]
Thus $A_0,b_0$ satisfy the hypotheses of Lemma~\ref{lem:rounding_certificate}.
Moreover, for every vector $w$,
\begin{equation}\label{eq:fp_scale_identity}
    A(2^{\beta-\alpha}w)-b=2^\beta(A_0w-b_0),
    \qquad \|b\|_2=2^\beta\|b_0\|_2.
\end{equation}
Consequently, a relative-residual guarantee for the normalized system transfers unchanged to the original system.

For a chosen integer precision $h\ge1$, set $S=2^h$ and form
\begin{equation}\label{eq:fp_integer_conversion}
    C=\operatorname{round}(SA_0),\qquad
    d=\operatorname{round}(Sb_0).
\end{equation}
All subsequent lifting arithmetic uses the exact integers $C,d$, not the original floating-point entries.
Equivalently, $C/S$ and $d/S$ are fixed-point approximations to $A_0$ and $b_0$ with entrywise error at most $1/(2S)$.
The entries of $C,d$ have magnitude at most $S$, hence $O(h)$ bits, and $C$ has at most $m$ nonzeros.
Only the original nonzero entries need to be visited.

More explicitly, an original matrix entry $A_{ij}=u2^e$ is converted to
\[
    C_{ij}=\operatorname{round}\!\left(u2^{e-\alpha+h}\right),
\]
and the corresponding formula for a right-hand-side entry replaces $\alpha$ by $\beta$.
Thus conversion uses a binary shift followed, if needed, by rounding to the nearest integer with a fixed tie rule.
For a normalized entry $u2^\eta$, let $k=\eta+h$.
If $k+\operatorname{len}(u)\le-1$, then $|u2^k|<1/2$ and the converted entry is zero; this is detected by comparing encoded exponents.
Otherwise $|k|=O(B_{\rm sig}+h)$: a nonnegative $k$ calls for a left shift, while a negative $k$ calls for division by $2^{-k}$ and rounding.
Note that no huge denominator is constructed for an entry far below the rounding threshold.
This is why the conversion cost depends on exponent encoding lengths, rather than on their numerical magnitudes.

\subsection{The algorithm with a known condition-number bound}\label{sec:fp_known_condition}
Suppose for now that an integer upper bound $\bar\kappa\ge\kappa_2(A)$ is known.
For the requested tolerance $\epsilon=2^{-E}$, choose
\begin{align*}
    \ell&=\lceil\log_2 n\rceil,&
    h_0&=6+3\ell+E,&
    h&=h_0+\lceil\log_2\bar\kappa\rceil,\\
    S&=2^h,&
    p&=E+2\ell+2,&
    P&=2^p.
\end{align*}
The input-rounding scale is $S$; the output-rounding scale is $P$.
Convert the normalized inputs to $C,d$ using~\eqref{eq:fp_integer_conversion}.
Set $R=2^{h+\ell+1}\ge2nS$, and make one exact integer-core call:
\begin{equation}\label{eq:fp_integer_solve}
    z=\textsc{SolvePerturbed}(C,d,R),\qquad y=Rz.
\end{equation}
Since $(I+RC)z=d$, the vector $y$ solves
\[
    \left(\frac CS+\frac{I}{RS}\right)y=\frac dS.
\]
In particular, $y$ is a candidate for the normalized solution, not for the original-scale solution; no additional factor of $S$ is needed.
Round its coordinates to the output grid and undo the normalization by returning
\begin{equation}\label{eq:fp_output}
    q_i=\operatorname{round}(Py_i),\qquad
    \gamma=\beta-\alpha-p,
    \qquad \widehat x_i=q_i2^\gamma.
\end{equation}
Indeed, if $\widehat y=q/P$, then $\widehat x=2^{\beta-\alpha}\widehat y$.
Only the integers $q_i$ and $\gamma$ are written; the power $2^\gamma$ is not expanded.

The prescribed precision gives $S\ge64\bar\kappa n^3/\epsilon$ and $P\ge4n^2/\epsilon$.
The final assertion of Lemma~\ref{lem:rounding_certificate} therefore guarantees that its certificate holds and that $\widehat y$ has normalized relative residual at most $\epsilon/2$.
Equation~\eqref{eq:fp_scale_identity} gives the same guarantee for $\widehat x$ against $A,b$.
Thus, with a valid supplied bound, the algorithm needs only one input conversion and one call to \textsc{SolvePerturbed}, and no precision search or acceptance test is necessary.

\subsection{Removing the condition-number bound}\label{sec:fp_unknown_condition}
Without $\bar\kappa$, retain the same normalization and output scale $P$, but start at $h=h_0$.
At each trial, form $S,C,d,R,z,y$ exactly as above using the current $h$.
Accept the trial if
\[
    8n^2(\|y\|_\infty+1)\le\epsilon S,
\]
which is the certificate~\eqref{eq:rounding_certificate_test}.
On acceptance return~\eqref{eq:fp_output}; otherwise replace $h$ by $2h$ and repeat.
The integer data $C,d$ are recomputed from the original encoded inputs at the new precision.
This repetition and the acceptance test are the only additions to the known-bound algorithm; neither $\kappa$ nor an estimate of it is computed.

The test uses only the exact reconstructed fractions, not products with the original floating-point data.
If $z_i=a_i/g_i$ with $g_i>0$, it is equivalent to the $n$ integer comparisons
\begin{equation}\label{eq:fp_integer_certificate}
    2^E\,8n^2(R|a_i|+g_i)\le S g_i
    \qquad(1\le i\le n).
\end{equation}
An insufficient precision can make $C$ singular, but every trial is still defined, since $I+RC$ is nonsingular, and Theorem~\ref{thm:solve_perturbed} explicitly permits singular $C$.
Note that the test is sufficient rather than necessary, since rejecting a candidate does not guarantee that its actual residual is large.

\subsection{Correctness and complexity}
\begin{proof}[Proof of Theorem~\ref{thm:main}]
The zero right-hand side was handled separately.
For every other input, Lemma~\ref{lem:rounding_certificate} and~\eqref{eq:fp_scale_identity} show that every accepted output has the required original relative residual.
A trial must pass once
\[
    h\ge h_0+\lceil\log_2\kappa_2(A)\rceil,
\]
since this ensures $S\ge64\kappa_2(A)n^3/\epsilon$.
The largest tested $h$ is therefore $O(\rho)$, the sum of all tested $h$ is $O(\rho)$, and there are $O(\log(\rho+1))$ trials.
These bounds use the actual condition number only in the analysis, not in the algorithm.

At precision $h$, Section~\ref{sec:fp_conversion} implements each input conversion using $O(B_{\rm sig}+h)$ actual significand bits.
Exponent arithmetic uses $O(B_{\rm exp}+\log(B_{\rm sig}+h+1))$ bits, so the cost per entry is $\Ot(B_{\rm sig}+B_{\rm exp}+h)$.
Computing the normalization exponents has the same input-linear cost.
Since $C$ has at most $m$ nonzeros and $\log_2R=h+\ell+1=O(h)$,
Theorem~\ref{thm:solve_perturbed} gives $\Ot((mn+n^2)h)=\Ot(mnh)$ work for each exact solve, using $m\ge n$ for the original nonsingular matrix.
The recovered numerators and denominators have $O(nh)$ bits.
As $E+\log(n+1)=O(h)$, all comparisons in~\eqref{eq:fp_integer_certificate} cost $\Ot(n^2h)$ bit operations.
On acceptance, forming each $q_i$ requires a shift and integer division on $O(nh)$-bit integers, again costing $\Ot(n^2h)$ in total.
Summing over trials gives
\[
    \Ot\!\left((m+n)(B_{\rm sig}+B_{\rm exp})+mn\rho\right)
\]
bit operations, and repeated input conversion contributes only a logarithmic factor to the input-reading term.

Retaining the original encoded inputs uses $O((m+n)(B_{\rm sig}+B_{\rm exp})+m\log(n+1))$ bits.
The converted data, digit vectors, reconstructed fractions, and scalar workspaces use $\Ot((m+n^2)h)=\Ot(n^2h)$ bits at one trial, and failed-trial storage can be reused.
The sparse-index term is covered by $n^2\rho$.
This proves the stated storage bound.

Finally, acceptance gives $\|y\|_\infty+1\le\epsilon2^h/(8n^2)$, and $|q_i|\le P\|y\|_\infty+1/2$.
Since $p=O(E+\log(n+1))$ and $h=O(\rho)$, each $q_i$ has $O(\rho)$ bits.
The definitions of $\alpha,\beta$ give
\[
    |\gamma|\le2^{B_{\rm exp}+1}+2B_{\rm sig}+p,
\]
so its encoding has $O(B_{\rm exp}+\log(B_{\rm sig}+\rho+1))$ bits.
Only this encoding, not the expanded value $2^\gamma$, is written.
\end{proof}

\paragraph{Fixed-point inputs.}
For $B$-bit fixed-point inputs, where the bit count includes the integer and fractional parts, take $B_{\rm sig}=O(B)$ and $B_{\rm exp}=O(\log(B+1))$.
The normalization exponents $\alpha,\beta$ have magnitude $O(B)$, so $|\gamma|=O(B+E+\log(n+1))$.
Thus the output can be expanded exactly into fixed-point form using $O(B+\rho)$ bits per coordinate and $O(n(B+\rho))$ additional bit operations.
In particular, with $\tau=\max\{B,\log\kappa,E\}$, this gives $\Ot(mn\tau)$ bit operations, $\Ot(n^2\tau)$ bits of storage, and $\Ot(\tau)$ bits per output coordinate.

Note that a compact output representation is necessary in general to avoid expanding very large powers of two.
For example, if $A=2^{-D}I$ and $b=e_1$, with $D$ encoded in binary, $\kappa_2(A)=1$ but an ordinary fixed-point output of relative residual less than $1/2$ needs $\Omega(D)$ bits in its first coordinate.

\section{Acknowledgments}
 The author would like to thank Richard Peng for reading and providing helpful comments about an early draft of this paper.
 \subsection{AI Disclosure}
 The author interacted with GPT 5.6 Sol and GPT 6 Astra at various points in this project to help explore and numerically test a long list of approaches to this problem. 
 The author also used these models and Claude Opus 5.5 during the writing, primarily to assist with literature search and to proofread and polish the exposition, but also to help write out some aspects of the rounding analysis for floating-point inputs (which he expected would be messier than they ended up being).	
 
\bibliographystyle{plain}
\bibliography{references}
\end{document}